\documentclass[12pt,oneside]{amsart}

\usepackage{amssymb}
\usepackage{amsfonts}
\usepackage{amscd}
\usepackage[matrix, arrow, curve]{xy}
\usepackage{graphicx}
\usepackage[cp1251]{inputenc}

\numberwithin{equation}{section}

\newtheorem{theorem}{Theorem}[section]

\newtheorem{prop}[theorem]{Proposition}
\newtheorem{corollary}[theorem]{Corollary}

\theoremstyle{definition}

\theoremstyle{remark}

\begin{document}
\newcommand{\M}{\mathcal{M}}
\newcommand{\F}{\mathcal{F}}
\newcommand{\mcG}{\mathcal{G}}
\newcommand{\mcN}{\mathcal{N}}
\newcommand{\mcK}{\mathcal{K}}

\newcommand{\Teich}{\mathcal{T}_{g,N+1}^{(1)}}
\newcommand{\T}{\mathrm{T}}
\newcommand{\Mer}{\mathrm{Mer}}
\newcommand{\corr}{\bf}
\newcommand{\vac}{|0\rangle}
\newcommand{\Ga}{\Gamma}
\newcommand{\new}{\bf}
\newcommand{\define}{\def}
\newcommand{\redefine}{\def}
\newcommand{\Cal}[1]{\mathcal{#1}}
\renewcommand{\frak}[1]{\mathfrak{{#1}}}
\newcommand{\Hom}{\rm{Hom}\,}
\newcommand{\refE}[1]{(\ref{E:#1})}
\newcommand{\refCh}[1]{Chapter~\ref{Ch:#1}}
\newcommand{\refS}[1]{Section~\ref{S:#1}}
\newcommand{\refSS}[1]{Section~\ref{SS:#1}}
\newcommand{\refT}[1]{Theorem~\ref{T:#1}}
\newcommand{\refO}[1]{Observation~\ref{O:#1}}
\newcommand{\refP}[1]{Proposition~\ref{P:#1}}
\newcommand{\refD}[1]{Definition~\ref{D:#1}}
\newcommand{\refC}[1]{Corollary~\ref{C:#1}}
\newcommand{\refL}[1]{Lemma~\ref{L:#1}}
\newcommand{\R}{\ensuremath{\mathbb{R}}}
\newcommand{\N}{\ensuremath{\mathbb{N}}}
\newcommand{\Q}{\ensuremath{\mathbb{Q}}}
\renewcommand{\P}{\ensuremath{\mathcal{P}}}
\newcommand{\Z}{\ensuremath{\mathbb{Z}}}
\newcommand{\kv}{{k^{\vee}}}
\renewcommand{\l}{\lambda}
\newcommand{\gb}{\overline{\mathfrak{g}}}
\newcommand{\dt}{\tilde d}     
\newcommand{\hb}{\overline{\mathfrak{h}}}
\newcommand{\g}{\mathfrak{g}}
\newcommand{\h}{\mathfrak{h}}
\newcommand{\gh}{\widehat{\mathfrak{g}}}
\newcommand{\ghN}{\widehat{\mathfrak{g}_{(N)}}}
\newcommand{\gbN}{\overline{\mathfrak{g}_{(N)}}}
\newcommand{\tr}{\mathrm{tr}}
\newcommand{\gln}{\mathfrak{gl}(n)}
\newcommand{\son}{\mathfrak{so}(n)}
\newcommand{\spnn}{\mathfrak{sp}(2n)}
\newcommand{\sln}{\mathfrak{sl}}
\newcommand{\sn}{\mathfrak{s}}
\newcommand{\so}{\mathfrak{so}}
\newcommand{\spn}{\mathfrak{sp}}
\newcommand{\tsp}{\mathfrak{tsp}(2n)}
\newcommand{\gl}{\mathfrak{gl}}
\newcommand{\slnb}{{\overline{\mathfrak{sl}}}}
\newcommand{\snb}{{\overline{\mathfrak{s}}}}
\newcommand{\sob}{{\overline{\mathfrak{so}}}}
\newcommand{\spnb}{{\overline{\mathfrak{sp}}}}
\newcommand{\glb}{{\overline{\mathfrak{gl}}}}
\newcommand{\Hwft}{\mathcal{H}_{F,\tau}}
\newcommand{\Hwftm}{\mathcal{H}_{F,\tau}^{(m)}}

\newcommand{\car}{{\mathfrak{h}}}    
\newcommand{\bor}{{\mathfrak{b}}}    
\newcommand{\nil}{{\mathfrak{n}}}    
\newcommand{\vp}{{\varphi}}
\newcommand{\bh}{\widehat{\mathfrak{b}}}  
\newcommand{\bb}{\overline{\mathfrak{b}}}  
\newcommand{\Vh}{\widehat{\mathcal V}}
\newcommand{\KZ}{Kniz\-hnik-Zamo\-lod\-chi\-kov}
\newcommand{\TUY}{Tsuchia, Ueno  and Yamada}
\newcommand{\KN} {Kri\-che\-ver-Novi\-kov}
\newcommand{\pN}{\ensuremath{(P_1,P_2,\ldots,P_N)}}
\newcommand{\xN}{\ensuremath{(\xi_1,\xi_2,\ldots,\xi_N)}}
\newcommand{\lN}{\ensuremath{(\lambda_1,\lambda_2,\ldots,\lambda_N)}}
\newcommand{\iN}{\ensuremath{1,\ldots, N}}
\newcommand{\iNf}{\ensuremath{1,\ldots, N,\infty}}

\newcommand{\tb}{\tilde \beta}
\newcommand{\tk}{\tilde \varkappa}
\newcommand{\ka}{\kappa}
\renewcommand{\k}{\varkappa}
\newcommand{\ce}{{c}}

\newcommand{\Pif} {P_{\infty}}
\newcommand{\Pinf} {P_{\infty}}
\newcommand{\PN}{\ensuremath{\{P_1,P_2,\ldots,P_N\}}}
\newcommand{\PNi}{\ensuremath{\{P_1,P_2,\ldots,P_N,P_\infty\}}}
\newcommand{\Fln}[1][n]{F_{#1}^\lambda}
\newcommand{\tang}{\mathrm{T}}
\newcommand{\Kl}[1][\lambda]{\can^{#1}}
\newcommand{\A}{\mathcal{A}}
\newcommand{\V}{\mathcal{V}}
\newcommand{\W}{\mathcal{W}}
\renewcommand{\O}{\mathcal{O}}
\newcommand{\Ae}{\widehat{\mathcal{A}}}
\newcommand{\Ah}{\widehat{\mathcal{A}}}
\newcommand{\La}{\mathcal{L}}
\newcommand{\Le}{\widehat{\mathcal{L}}}
\newcommand{\Lh}{\widehat{\mathcal{L}}}
\newcommand{\eh}{\widehat{e}}
\newcommand{\Da}{\mathcal{D}}
\newcommand{\kndual}[2]{\langle #1,#2\rangle}
\newcommand{\cins}{\frac 1{2\pi\mathrm{i}}\int_{C_S}}
\newcommand{\cinsl}{\frac 1{24\pi\mathrm{i}}\int_{C_S}}
\newcommand{\cinc}[1]{\frac 1{2\pi\mathrm{i}}\int_{#1}}
\newcommand{\cintl}[1]{\frac 1{24\pi\mathrm{i}}\int_{#1 }}
\newcommand{\w}{\omega}
\newcommand{\ord}{\operatorname{ord}}
\newcommand{\res}{\operatorname{res}}
\newcommand{\nord}[1]{:\mkern-5mu{#1}\mkern-5mu:}
\newcommand{\codim}{\operatorname{codim}}
\newcommand{\ad}{\operatorname{ad}}
\newcommand{\Ad}{\operatorname{Ad}}

\newcommand{\Fn}[1][\lambda]{\mathcal{F}^{#1}}
\newcommand{\Fl}[1][\lambda]{\mathcal{F}^{#1}}
\renewcommand{\Re}{\mathrm{Re}}

\newcommand{\ha}{H^\alpha}

\define\ldot{\hskip 1pt.\hskip 1pt}
\define\ifft{\qquad\text{if and only if}\qquad}
\define\a{\alpha}
\redefine\d{\delta}
\define\w{\omega}
\define\ep{\epsilon}
\redefine\b{\beta} \redefine\t{\tau} \redefine\i{{\,\mathrm{i}}\,}
\define\ga{\gamma}
\define\S{\Sigma}
\define\cint #1{\frac 1{2\pi\i}\int_{C_{#1}}}
\define\cintta{\frac 1{2\pi\i}\int_{C_{\tau}}}
\define\cintt{\frac 1{2\pi\i}\oint_{C}}
\define\cinttp{\frac 1{2\pi\i}\int_{C_{\tau'}}}
\define\cinto{\frac 1{2\pi\i}\int_{C_{0}}}
\define\cinttt{\frac 1{24\pi\i}\int_C}
\define\cintd{\frac 1{(2\pi \i)^2}\iint\limits_{C_{\tau}\,C_{\tau'}}}
\define\dintd{\frac 1{(2\pi \i)^2}\iint\limits_{C\,C'}}
\define\cintdr{\frac 1{(2\pi \i)^3}\int_{C_{\tau}}\int_{C_{\tau'}}
\int_{C_{\tau''}}}
\define\im{\operatorname{Im}}
\define\re{\operatorname{Re}}
\define\res{\operatorname{res}}
\redefine\deg{\operatornamewithlimits{deg}}
\define\ord{\operatorname{ord}}
\define\rank{\operatorname{rank}}
\define\fpz{\frac {d }{dz}}
\define\dzl{\,{dz}^\l}
\define\pfz#1{\frac {d#1}{dz}}

\define\K{\Cal K}
\define\U{\Cal U}
\redefine\O{\Cal O}
\define\He{\text{\rm H}^1}
\redefine\H{{\mathrm{H}}}
\define\Ho{\text{\rm H}^0}
\define\A{\Cal A}
\define\Do{\Cal D^{1}}
\define\Dh{\widehat{\mathcal{D}}^{1}}
\redefine\L{\Cal L}
\newcommand{\ND}{\ensuremath{\mathcal{N}^D}}
\redefine\D{\Cal D^{1}}
\define\KN {Kri\-che\-ver-Novi\-kov}
\define\Pif {{P_{\infty}}}
\define\Uif {{U_{\infty}}}
\define\Uifs {{U_{\infty}^*}}
\define\KM {Kac-Moody}
\define\Fln{\Cal F^\lambda_n}
\define\gb{\overline{\mathfrak{ g}}}
\define\G{\overline{\mathfrak{ g}}}
\define\Gb{\overline{\mathfrak{ g}}}
\redefine\g{\mathfrak{ g}}
\define\Gh{\widehat{\mathfrak{ g}}}
\define\gh{\widehat{\mathfrak{ g}}}
\define\Ah{\widehat{\Cal A}}
\define\Lh{\widehat{\Cal L}}
\define\Ugh{\Cal U(\Gh)}
\define\Xh{\hat X}
\define\Tld{...}
\define\iN{i=1,\ldots,N}
\define\iNi{i=1,\ldots,N,\infty}
\define\pN{p=1,\ldots,N}
\define\pNi{p=1,\ldots,N,\infty}
\define\de{\delta}

\define\kndual#1#2{\langle #1,#2\rangle}
\define \nord #1{:\mkern-5mu{#1}\mkern-5mu:}
\newcommand{\MgN}{\mathcal{M}_{g,N}} 
\newcommand{\MgNeki}{\mathcal{M}_{g,N+1}^{(k,\infty)}} 
\newcommand{\MgNeei}{\mathcal{M}_{g,N+1}^{(1,\infty)}} 
\newcommand{\MgNekp}{\mathcal{M}_{g,N+1}^{(k,p)}} 
\newcommand{\MgNkp}{\mathcal{M}_{g,N}^{(k,p)}} 
\newcommand{\MgNk}{\mathcal{M}_{g,N}^{(k)}} 
\newcommand{\MgNekpp}{\mathcal{M}_{g,N+1}^{(k,p')}} 
\newcommand{\MgNekkpp}{\mathcal{M}_{g,N+1}^{(k',p')}} 
\newcommand{\MgNezp}{\mathcal{M}_{g,N+1}^{(0,p)}} 
\newcommand{\MgNeep}{\mathcal{M}_{g,N+1}^{(1,p)}} 
\newcommand{\MgNeee}{\mathcal{M}_{g,N+1}^{(1,1)}} 
\newcommand{\MgNeez}{\mathcal{M}_{g,N+1}^{(1,0)}} 
\newcommand{\MgNezz}{\mathcal{M}_{g,N+1}^{(0,0)}} 
\newcommand{\MgNi}{\mathcal{M}_{g,N}^{\infty}} 
\newcommand{\MgNe}{\mathcal{M}_{g,N+1}} 
\newcommand{\MgNep}{\mathcal{M}_{g,N+1}^{(1)}} 
\newcommand{\MgNp}{\mathcal{M}_{g,N}^{(1)}} 
\newcommand{\Mgep}{\mathcal{M}_{g,1}^{(p)}} 
\newcommand{\MegN}{\mathcal{M}_{g,N+1}^{(1)}} 

\define \sinf{{\widehat{\sigma}}_\infty}
\define\Wt{\widetilde{W}}
\define\St{\widetilde{S}}
\newcommand{\SigmaT}{\widetilde{\Sigma}}
\newcommand{\hT}{\widetilde{\frak h}}
\define\Wn{W^{(1)}}
\define\Wtn{\widetilde{W}^{(1)}}
\define\btn{\tilde b^{(1)}}
\define\bt{\tilde b}
\define\bn{b^{(1)}}
\define \ainf{{\frak a}_\infty} 

%
\define\eps{\varepsilon}    
\newcommand{\e}{\varepsilon}
\define\doint{({\frac 1{2\pi\i}})^2\oint\limits _{C_0}
       \oint\limits _{C_0}}                            
\define\noint{ {\frac 1{2\pi\i}} \oint}   
\define \fh{{\frak h}}     
\define \fg{{\frak g}}     
\define \GKN{{\Cal G}}   
\define \gaff{{\hat\frak g}}   
\define\V{\Cal V}
\define \ms{{\Cal M}_{g,N}} 
\define \mse{{\Cal M}_{g,N+1}} 
\define \tOmega{\Tilde\Omega}
\define \tw{\Tilde\omega}
\define \hw{\hat\omega}
\define \s{\sigma}
\define \car{{\frak h}}    
\define \bor{{\frak b}}    
\define \nil{{\frak n}}    
\define \vp{{\varphi}}
\define\bh{\widehat{\frak b}}  
\define\bb{\overline{\frak b}}  
\define\KZ{Knizhnik-Zamolodchikov}
\define\ai{{\alpha(i)}}
\define\ak{{\alpha(k)}}
\define\aj{{\alpha(j)}}
\newcommand{\calF}{{\mathcal F}}
\newcommand{\ferm}{{\mathcal F}^{\infty /2}}
\newcommand{\Aut}{\operatorname{Aut}}
\newcommand{\End}{\operatorname{End}}
\newcommand{\laxgl}{\overline{\mathfrak{gl}}}
\newcommand{\laxsl}{\overline{\mathfrak{sl}}}
\newcommand{\laxso}{\overline{\mathfrak{so}}}
\newcommand{\laxsp}{\overline{\mathfrak{sp}}}
\newcommand{\laxs}{\overline{\mathfrak{s}}}
\newcommand{\laxg}{\overline{\frak g}}
\newcommand{\bgl}{\laxgl(n)}
\newcommand{\zs}{\overline{z}}
\newcommand{\tX}{\widetilde{X}}
\newcommand{\tY}{\widetilde{Y}}
\newcommand{\tZ}{\widetilde{Z}}


\title[]{Separation of variables for type $D_n$ Hitchin systems on hyperelliptic curves}
\author[P.I. Borisova]{P.I. Borisova}
\maketitle
\tableofcontents
\begin{abstract}
We use separation of variables to find explicit formulas for the Darboux coordinates for the Hitchin systems on a hyprelliptic curve for simple Lie algebras of type $D_n$.
\end{abstract}

\section*{Introduction}
The action-angle coordinates for the Hitchin systems were known only in case of  Lie algebra $\sln(2)$ and a curve of genus 2 (\cite{Gw}), until recently. A description of the class of spectral curves for the Hitchin systems on
hyperelliptic curves of any rank was given in \cite{Sheinman_arx}. For Lie algebras of type $A_n, B_n, C_n$ the Darboux coordinates were found there in the explicit form.

The main goal of this paper is to find the Darboux coordinates for the Hitchin systems on a hyprelliptic curve for simple Lie algebras of type $D_n$, according to explicit description of the spectral curve, given in \cite{Sheinman_arx}. Nevertheless, we were able to explicitly obtain the Darboux coordinates only for the simplest case of Lie algebra $\so(4)$. Note that the isomorphism $\so(4)\cong \sln(2)\oplus\sln(2)$ doesn't give any result, because this isomorphism is an outer isomorphism and it doesn't preserve the spectral curve.

Consider a Hitchin system on a hyprelliptic curve $\Sigma_g$ of genus $g$, given by $y^2 =~P_{2g+1}(x)$, for a classical Lie algebra $\g$. This system is defined by the Lax operator $L$ ( \cite{Krichever_CMP}, \cite{Sheinman_arx},\cite{Kr_Sh_FAN}). The spectral curve is given by the equation $\det(\l - L)=0$. In terms of $x,y$   it looks like:
\begin{equation}\label{E:eq0} R(\lambda, x, y, H)=\lambda^{n}+\sum_{i=1}^{n}\lambda^{n-i}r_{i}(x, y, H)=0,
\end{equation}
where $n$ is the dimension of the standard representation of the Lie algebra $\g$, $r_{i}(x, y, H)$ are meromorphic functions on $\Sigma_g$ such that $r_i(x, y) = \chi_i(L(x, y))$, $\chi_i~(i = 1, \ldots, n)$ are invariant polynomials of the Lie algebra $\g$, $H$ are Hamiltonians of the system (which arise as coefficients of the expansion of those invariant polynomials over a basis of meromorphic functions).


\section{Hamiltonians}

Consider the Hitchin system for a simple Lie algebra $\g$ of type $D_n$. This algebra $\g$ has $n$ fundamental invariants of orders $ \{2, 4, \ldots, 2(n - 1), n \}$ and the spectral curve is given by the following equation
\[ R(\lambda, x, y, H)=\lambda^{2n}+\sum_{i=1}^{n - 1}\lambda^{2n-2i}r_{i}(x, y, H) + (r_{0}(x, y, H))^2=0,
\]
where $r_i(x, y, H)$ is a fundamental invariant of degree $2i, ~(i = 1, \ldots, (n-1))$, $r_0(x, y, H)$ is the invariant of degree $n$ (Pfaffian).

According to \cite{Sheinman_arx}, a fundamental invariant $r_i$  of degree $d_i$ can be expanded over the basis formed by two series of functions: $\{1, x, \ldots, x^{d_i(g - 1)}\}$ and $\{y, yx, \ldots, yx^{(d_i - 1)(g - 1) - 2}\}$ as follows:
\begin{equation}\label{E:eq1} r_{i}(x, y, H) = \sum\limits_{k=0}^{d_i(g - 1)}H^{(0)}_{ik}x^k + \sum\limits_{s = 0}^{(d_i - 1)(g - 1) - 2}H^{(1)}_{is}yx^s.
\end{equation}

Consider the Hitchin system for the Lie algebra $\so(4)$ on a hyperelliptic curve of genus ~$2$.
The dimension of the space of spectral invariants is $ N = \dim\g\cdot(g - 1) = 6$.
Lie algebra $\so(4)$ has two fundamental invariants of degree $2$, which are expanded over the  basis $\{1, x, x^2\}$, according to \refE{eq1}. Hence, the explicit equation for the spectral curve is
\begin{equation}\label{E:spec}
    \lambda^4 + (H_4 + xH_5 + x^2H_6)\lambda^2 + (H_1 + xH_2 + x^2H_3)^2 = 0.
\end{equation}

We can find the Hamiltonians using the fact that the spectral curve passes through points $(\lambda_i, x_i)$. In other words, our problem is reduced to solving a system of $6$ equations in variables $H = \{ H_j |j = 1, \ldots, 6\}$. This leads us to the following proposition.
\begin{prop}\label{P:pr1}
System for $\g = \so(4)$ can be reduced to an equation of degree $4$ in one variable and, as a consequence, is solvable in radicals.
\end{prop}

\begin{proof}
 Performing Gaussian elimination, one can eliminate variables $H_4$, $H_5$, $H_6$. Thus, the system is reduced to a system of three homogeneous equations with a nonzero right part:
\[a_{i1}H_1^2+a_{i2}H_2^2+a_{i3}H_3^2+b_{i1}H_1H_2+b_{i2}H_1H_3+b_{i3}H_2H_3 = \Tilde{\l}_i,~ i = 1, 2, 3.
\]
Consider a matrix $\{a_{ij}\}_{i, j = 1, 2, 3}$ of the coefficients of quadratic terms. Performing Gaussian elimination again we obtain a diagonal matrix. Dividing the second two equations by $\frac{\l_1}{\l_2}$ and $\frac{\l_1}{\l_3}$ respectively, we obtain the following system.
$$
\left \{
\begin{array}{c}
    \tilde{a}_{11}H_1^2+\tilde{b}_{11}H_1H_2+\tilde{b}_{12}H_1H_3+\tilde{b}_{13}H_2H_3 = \Tilde{\l}_1 \\
    \tilde{a}_{22}H_2^2+\tilde{b}_{21}H_1H_2+\tilde{b}_{22}H_1H_3+\tilde{b}_{23}H_2H_3 = \Tilde{\l}_1 \\
    \tilde{a}_{33}H_3^2+\tilde{b}_{31}H_1H_2+\tilde{b}_{32}H_1H_3+\tilde{b}_{33}H_2H_3 = \Tilde{\l}_1
\end{array}
\right.
$$
On the next step, subtracting the first equation from the second two, we get a system of two equations with a nonzero right part. Next, we divide both sides by $H_1$ to eliminate one variable. After all the before-mentioned, we get following system:
\[\tilde{a}_{i1}\tilde{H}_2^2+\tilde{a}_{i2}\tilde{H}_3^2+c_{i1}\tilde{H}_2\tilde{H}_3+c_{i2}\tilde{H}_2+c_{i3}\tilde{H}_3 = \tilde{\tilde{\l}}_i,~ i = 1, 2, ~ \tilde{a}_{11} = 0,~\tilde{a}_{22} = 0.
\]
Finally, we can express $\tilde{H}_2$ via $\tilde{H}_3$ from the first equation to get the equation of degree $4$ in one variable $\tilde{H}_3$. This equation is solvable in radicals. Thus we obtain the explicit expression for the action coordinates for the Hitchin system for $D_2$.
\end{proof}

In the general case of $D_n,~ n > 2$, according to the before-mentioned method, we can reduce original system to a system of homogeneous equations with nonzero right part in $(2n-1)(g-1)$ variables. In general, such systems are not solvable in radicals (\cite{Esterov}). In this case, the Hamiltonians of the Hitchin system are given by implicit formulas $R(\l_i, x_i, y_i, H) = 0, ~ i = 1, \ldots, N$.


\section{Angle coordinates}

Recall from \cite{Sheinman_arx} that the symplectic form on a phase space of the Hitchin system of any rank on a hyperelliptic curve is given by the formula:
\begin{equation}\label{E:sympl}
\s = \sum\limits_{i}^{}d\l_i\wedge\frac{dx_i}{y_i},
\end{equation}
where $\frac{dx}{y}$ is a holomorphic differential on the hyperelliptic curve. For further calculations we can bring this form to the canonical one via the appropriate change of coordinates
\[(\l_i, x_i, y_i)\longmapsto(\l_i,\tilde{x}_i, y_i), ~\tilde{x}_i =\int\limits^{(x_i, y_i)}\frac{dx}{y}.\]

Furthermore, the following formula for coordinates $\{\vp_j\}$, which are conjugate with coordinates $\{H_j\}$, can be obtained from \cite[(1.7)]{Hurtubise1}:
\begin{equation}\label{E:angle}
\vp_j =  - \sum\limits_{i = 1}^{N}\int\limits_{}^{(x_i, y_i)}\frac{R'_{H_j}(\lambda,x, y, H)}{yR'_{\l}(\lambda, x, y, H)}dx.
\end{equation}
(the $y$ in denominator appears due to the fact that our coordinates are not in a canonical form first and we have to do a change of variables first). Similar, but more specific formulas were given in \cite{Hurtubise1}, \cite{Talalaev}, \cite{Sklyanin}.

To obtain formula \refE{angle} one can use the method of separation of variables.
We have a system of $N = \dim\g\cdot(g - 1)$ equations $R_i(\l_i, x_i, y_i, H) = 0$ (the  equations  of  the  spectral  curve  passing  through $N$ points $(\l_i, x_i, y_i)$). Suppose that $\frac{\partial S}{\partial \tilde{x}_i} = r(\tilde{x}_i, H)$, hence:
\[S =\sum_{i = 1}^{N} \int\limits_{}^{(x_i, y_i)}r(\tilde{x}, H)d\tilde{x} = \sum_{i = 1}^{N}\int\limits_{}^{(x_i, y_i)}r(x, H)\frac{dx}{y}.
\]
So we can put $\frac{\partial S}{\partial x_i}$ in the spectral curve equation instead of $\l_i$ and find the coordinates $\vp_j = \frac{\partial S}{\partial H_j}$ from the expression $\frac{d}{dH_j}(R(r(x, H), x, y, H)) = 0$. In particular,
\[\frac{\partial R}{\partial \l}\frac{\partial r}{\partial H_j} + \frac{\partial R}{\partial H_j} = 0.\]
And find the coordinates $\{\vp_j\}$ from the equation:
\begin{equation}\label{E:eq3}
\vp_j =  - \sum\limits_{i = 1}^{N}\int\limits_{}^{(x_i, y_i)}\frac{\partial r}{\partial H_j}\frac{dx}{y} = - \sum\limits_{i = 1}^{N}\int\limits_{}^{(x_i, y_i)}\frac{R'_{H_j}(\l, x, y, H)}{R'_{\l}(\l, x, y, H)}\frac{dx}{y}.
\end{equation}

According to the formula \refE{angle} and to the explicit formulas for the spectral curve \refE{eq0} and \refE{eq1} we can obtain the coordinates $\{\vp_j\}$ for the Hitchin systems for any classical system of roots on a hyperelliptic curve of any genus (see \cite{Sheinman_arx} for the root systems $A_n$, $B_n$, $C_n$). In particular,
\begin{corollary}
For Lie algebra $\so(4)$ the coordinates $\{\vp_j\}$ are of the form \refE{eq3}, where
\begin{align*}
       &R'_{H_j}(\lambda,x, y, H) =x^{j-1}(H_1 + xH_2 + x^2H_3),~ j = 1, 2, 3,\\
       & R'_{H_j}(\lambda,x, y, H) =\l^2x^{j-4},~ j = 4, 5, 6,\\
       &R'_{\l}(\lambda,x, y, H) =4\lambda^3+2\lambda(H_4 + xH_5 + x^2H_6).
\end{align*}
\end{corollary}

Author thanks O.K. Sheinman for the formulation of the problem and discussing, A. I. Esterov for explanations on solvability of systems of quadratic equations and I.M. Krichever for remarks about action-angle variables.

\end{document}